\documentclass[9pt,twocolumn,journal]{IEEEtran}
\usepackage{amssymb,amsfonts}
\usepackage[cmex10]{amsmath}
\usepackage{enumerate}
\usepackage{graphicx}
\usepackage{cite}

\usepackage{caption,subcaption}

\usepackage{tikz}
\usetikzlibrary{calc,shapes,arrows}

\def\bb#1{{\mathbb #1}}
\newcommand{\R}{{\rm I\!R}}
\newcommand{\dfb}{\stackrel{\Delta}{=}}

\def\scr#1{{\mathcal #1}}

\def\ddt{\frac{d}{dt}}
\def\la{\lambda}

\def\vp{\varphi}
\def\ve{\varepsilon}

\def\r{\mathbb R}

\def\G{\mathcal G}

\def\om{\omega}

\def\be{\beta}

\def\be{\begin{equation}}
\def\ee{\end{equation}}
\def\ben{\begin{equation*}}
\def\een{\end{equation*}}

\newtheorem{lemma}{Lemma}
\newtheorem{theorem}{Theorem}

\newtheorem{assumption}{Assumption}

\newtheorem{defn}{Definition}
\newtheorem{remark}{Remark}

\begin{document}
\title{Synchronization of Goodwin's oscillators under boundedness and nonnegativeness constraints for solutions.
}

\author{Anton~V.~Proskurnikov,~\IEEEmembership{Member,~IEEE}, Ming~Cao,~\IEEEmembership{Member,~IEEE}%
\thanks{The work was supported in part by the European Research Council (ERCStG-
307207), St. Petersburg State University, grant 6.38.230.2015 and RFBR grant 14-08-01015.
Theorems~\ref{thm3} and \ref{thm-my} were obtained, respectively, under sole support of Russian Science Foundation grants 16-19-00057 and 14-29-00142 
at Institute for Problems of Mechanical Engineering of Russian Academy of Sciences (IPME RAS).}.
\thanks{Anton V. Proskurnikov is with ENTEG at the Faculty of Mathematics and Natural Sciences, University of Groningen, Groningen,
the Netherlands, and also with St. Petersburg State University, ITMO University and IPME RAS, St.Petersburg, Russia;\;{\tt\small anton.p.1982@ieee.org}}%
\thanks{Ming Cao is with ENTEG at the Faculty of Mathematics and Natural Sciences, University of Groningen, The Netherlands;\;{\tt\small m.cao@rug.nl}}%
}

\maketitle

\begin{abstract}
In the recent paper by Hamadeh et al. (2012) an elegant analytic criterion for incremental output feedback passivity (iOFP) of cyclic feedback systems (CFS) has been reported,
assuming that the constituent subsystems are incrementally output strictly passive (iOSP). This criterion was used to prove that a network of identical CFS can be synchronized under
sufficiently strong linear diffusive coupling. A very important class of CFS consists of biological oscillators, named after Brian Goodwin and describing self-regulated chains of enzymatic reactions, 
where the product of each reaction catalyzes the next reaction, while the last product inhibits the first reaction in the chain. Goodwin's oscillators are used, in particular,
to model the dynamics of genetic circadian pacemakers, hormonal cycles and some metabolic pathways.

In this paper we point out that for Goodwin's oscillators, where the individual reactions have nonlinear (e.g. Mikhaelis-Menten) kinetics,
the synchronization criterion, obtained by Hamadeh et al., cannot be directly applied. This criterion relies on the implicit assumption of the solution
\emph{boundedness}, dictated also by the chemical feasibility (the state variables stand for the concentrations of chemicals).
Furthermore, to test the synchronization condition one needs to know an \emph{explicit} bound for a solution, which generally cannot be guaranteed under linear coupling.
At the same time, we show that these restrictions can be avoided for a \emph{nonlinear} synchronization protocol, where the control inputs are ``saturated''
by a special nonlinear function (belonging to a wide class), which guarantees nonnegativity of the solutions and allows to get explicit ultimate bounds for them.
We prove that oscillators synchronize under such a protocol, provided that the couplings are sufficiently strong.
\end{abstract}

\section{Introduction}

The rhythmicity of many vital processes in living organisms, such as the cell division,  blood pulse and breathing, diurnal sleep and wake cycle, are controlled by genetic and other biochemical ``clocks'', or \emph{pacemakers}, that are typically described by nonlinear systems of differential equations with stable limit cycles as their solutions. One of the first and most influential models of this type, describing genetic oscillators \cite{Murray02,GoBeWaKrHe05,LoWeKrHe08,VaHe11}, metabolic pathways in a cell \cite{Costalat96}
and hormonal cycles \cite{Murray02,Smith83,Medvedev09}, is known (along with its extensions) as  \emph{Goodwin's oscillator}. For 50 years since Goodwin's
seminal paper \cite{Goodwin65} this model has been attracting intensive attention in applied mathematics.

A challenging problem concerned with biochemical oscillators is to study mechanisms of their synchronization via coupling. Experiments and extensive simulations
(see \cite{GoBeWaKrHe05,LoWeKrHe08,VaHe11} and references therein) show that the stable 24h-periodic circadian rhythm is not inherent to intracellular genetic oscillators
(whose natural periods are spread from 20h to 28h) but emerges due to the \emph{coupling} among them, which also facilitates the oscillations' entrainability by
the daylight and other environmental cues (``zeitgebers'').

The Goodwin oscillator is a special case of a \emph{cyclic feedback system} (CFS), consisting of incrementally output passive blocks.
An important step in understanding the synchronization mechanism for such systems has been done in the recent paper \cite{HaStSeGo12}, establishing
an elegant criterion for synchronization of \emph{identical} CFS under sufficiently ``strong'' linear diffusive couplings. An ensemble of
CFS gets synchronized if the algebraic connectivity of the (weighted) digraph, describing the coupling between the systems,
exceeds the incremental passivity gain of the CFS. A critical observation is that this gain depends on the secant gains of the constituent subsystems.

As will be shown, the criterion from \cite{HaStSeGo12} not only adopts an implicit assumption of the solution's \emph{boundedness}, but in fact requires to find the bounds \emph{explicitly}. In order to apply this criterion, one needs to estimate the incremental passivity (or secant) gains of the blocks constituting the CFS. Synchronization is guaranteed only when these gains are finite, except for that of the leading block, since otherwise the minimal coupling strength, required to synchronize oscillators, becomes infinite.
As will be discussed in Section~\ref{sec.iofp}, the chemical reactions with linear kinetics correspond to the blocks with finite secant gains.
However, nonlinear (e.g. Mikhaelis-Menten) kinetics, typically arising in models of enzymatic and other biochemical reactions \cite{ChenNiepelSorger,OBrien2012},
lead to the \emph{infinite} passivity gain of the correspondent block. This gain becomes finite only for solutions, confined to some \emph{bounded} set, and to estimate the gain,
one has to find this set or, equivalently, explicit bound for the solution. For a linear diffusive protocol establishing such bounds is a non-trivial, and in fact open problem.
This hinders application of the criterion from \cite{HaStSeGo12} to Goodwin's biochemical oscillators with Mikhaelis-Menten nonlinearities, modeling e.g. the genetic circadian clocks \cite{GoBeWaKrHe05,LoWeKrHe08,VaHe11}.

In this technical note, we propose a modification of the algorithm from \cite{HaStSeGo12}, combining the usual diffusive coupling with a nonlinear ``saturating'' map,
which, similar to the linear protocol from \cite{HaStSeGo12}, guarantees non-negativity of the solutions but, additionally, provides an explicit upper
bound for the solutions. Under some technical assumptions, we prove that the ensemble of CFS's synchronizes, and find explicitly the margin for the coupling strength.
Unlike linear coupling protocols, the ``saturated'' protocol also guarantees non-negative control input which can be important when
such an input stands for some chemical concentration (e.g., the models of circadian oscillators from \cite{GoBeWaKrHe05,LoWeKrHe08,VaHe11} treat the input as the concentration of
a neurotransmitting polypeptide in the extracellular domain).

The main contribution of the paper is twofold. First, we point out some limitations of the synchronization criterion from \cite{HaStSeGo12}, concerned with the necessity to prove the solution
boundedness and estimate the incremental passivity gains. Second, we develop the approach from \cite{HaStSeGo12} to address ``saturated'' protocols,
providing synchronization of Goodwin-type oscillators with nonlinear reactions' kinetics. Dealing with a more general class of Goodwin's oscillators, our result inevitably inherits two basic 
limitations of the incremental passivity approach \cite{StanSepulchre,HaStSeGo12} and is confined to \emph{identical} oscillators and \emph{diffusive} couplings 
(the input of each oscillator depends on the deviation between its own and neighbors' outputs). The results of this paper can be applied e.g. to synchronization of \emph{synthetic} 
oscillator networks (see e.g. \cite{OBrien2012} and references therein), where individual oscillators and synchronization protocols are artificially engineered.

Note that such oscillator networks as the main circadian pacemaker in mammals consist of heterogeneous cells that are coupled non-diffusively (being, in fact, an example of pulse-coupled
network). A simplified continuous-time model for such a network, proposed in \cite{GoBeWaKrHe05,LoWeKrHe08}, employs a non-diffusive \emph{mean-field} coupling. Unlike the diffusive protocols,
under mean-field coupling the inputs of oscillators are \emph{identical} (depending on the average concentration of neurotransmitter, released by individual cells) and do not vanish as the oscillators get synchronized. Synchronization of oscillators under mean-field couplings and more complicated ``nearest-neighbor'' coupling rules \cite{VaHe11} remains a non-trivial mathematical problem.

\section{A class of cyclic feedback systems with incremental passivity properties}\label{sec.prelim}

We first briefly recall the central concepts of the incremental output \emph{strict} and \emph{feedback} passivity (iOSP and iOFP) \cite{HaStSeGo12}.
To simplify matters, we confine ourselves to single input-single output systems. We will make use the following notations. Given a vector $x=(x_1,\ldots,x_n)^{\top}\in\R^n$,
we denote $\|x\|\dfb\sqrt{x^{\top}x}$ and $\|x\|_{\infty}\dfb\max_i |x_i|$. We define ${\mathbf 1}_N\dfb(1,1,\ldots,1)\in\R^N$, and use the symbol $I_N$ for the identity $N\times N$-matrix. Given a matrix $L$, let $\|L\|_{\infty}\dfb \max_j\sum_{k}|L_{jk}|$ stand for its max norm; for $z=Lx$ one has then $\|z\|_{\infty}\le \|L\|_{\infty}\|x\|_{\infty}$.
Let $\R_+\dfb (0;+\infty)$ and $\overline\R_+\dfb [0;\infty)$.

\subsection{Definition of the iOFP property}

\begin{defn}\label{def.iosp} A system $\bb H$, whose dynamics obey
\begin{equation}\label{block1}
\bb H : \left \{
\begin{array}{lll}
 \dot x&=&\varphi(x,u) \\
y&=&\varrho(x,u)
\end{array}
\right.
\end{equation}
where $x\in\R^{d}$, $u\in\R$, and $y\in\R$ stand respectively for the state, input and output of $\bb H$, is said to be \emph{incrementally output strictly passive} with passivity (or \emph{secant}) gain $\gamma>0$, written iOSP($\gamma^{-1}$),  if a radially unbounded, positive definite function $S: \R^{r} \mapsto \R$ exists such that for any two solutions to (\ref{block1}), denoted repectively by $x^+$ associated with $y^+,u^+$ and $x^{\dagger}$ associated with $y^{\dagger},u^{\dagger}$, the increments $\Delta x=x^+-x^{\dagger}$, $\Delta y = y^+-y^{\dagger}$ and $\Delta u  =u^+-u^{\dagger}$ satisfy
\be\label{iOSP}
\ddt S(\Delta x)\le\Delta u\Delta y-\gamma^{-1}|\Delta y|^2.
\ee
The function $S$ is referred to as the \emph{incremental storage function}.
More generally,  $\bb H$ is said to be \emph{incrementally output feedback passive}, written  iOFP($\gamma^{-1}$), if inequality \eqref{iOSP} holds for some nonzero $\gamma\in\R\cup\{+\infty\}$ that has been relaxed from being strictly positive.
\end{defn}

In a degenerate case of \emph{static} input-out map of $\bb H$, taking the special form $u\mapsto y=g(u)$, inequality (\ref{iOSP}) simplifies to the condition
\be\label{iOSP0}
0\le\Delta u\Delta y-\frac{1}{\gamma}|\Delta y|^2,
\ee
i.e. $g$ is non-decreasing and Lipschitz: $|g(u_1)-g(u_2)|\le\gamma|u_1-u_2|$.

Definition~\ref{def.iosp} deals with the case when system $\bb H$ defined \emph{globally}, that is, $x(t)\in\R^{r}$ may be arbitrary.
The dynamics of biochemical systems are naturally defined in the positive orthant, and the iOFP property for such systems often can be proved in even more narrow domains.
We say the inequality \eqref{iOSP} is satisfied in a set $\G_x\subseteq\R^r$, if it holds for any two solutions $x^+$ associated with $y^+,u^+$ and $x^{\dagger}$ associated
with $y^{\dagger},u^{\dagger}$ as long as  $x^+(t)$ and $x^{\dagger}(t)$ have the property that $x^+(t),x^{\dagger}(t)\in \mathcal G_x$ for any $t\ge 0$. Here $\G_x$ is not necessarily invariant and we are only checking those solutions that stay in $\G_x$ for all $t$; the maps $\varphi,\varrho$ are defined on $\G_x\times\R^m$ and $S$ in \eqref{iOSP} is defined at least on the set $\G_x-\G_x \dfb \{x_1-x_2\in\R^{r}:x_1,x_2\in\mathcal G_x\}$, positive definite and radially unbounded (if $\G_x-\G_x$ is unbounded).
We call such a system $\bb H$ \emph{incrementally output feedback passive with gain $\gamma$ in the set $\G_x$}, written iOFP($\gamma^{-1}$,$\G_x$); if $\gamma>0$, we
call this property incrementally output \emph{strict passivity} in $\G_x$, written iOSP($\gamma^{-1}$,$\G_x$).

\subsection{An iOFP criterion for CFS}\label{sec.iofp}

Many \emph{cyclic feedback systems} (CFS), including Goodwin-type oscillators, appear to be iOSP or iOFP, provided that all of their sub-systems are iOSP. The results of \cite{HaStSeGo12} are concerned with CFS whose structures are described by the block diagram in
Fig. \ref{cyclic}. As illustrated, the overall system with the input $u_{ext}\in \R$ and the output $y_1\in \R$ consists of $n>1$ nonlinear subsystems $\bb H_i$, governed by
\begin{equation}\label{block}
\bb H_i : \left \{
\begin{array}{lll}
 \dot x_i&=&\varphi_i(x_i,u_i) \\
y_i&=&\varrho_i(x_i,u_i)
\end{array}
\right.,\quad i = 1, \ldots, n.
\end{equation}
Here $x_i\in\R^{r_i}$, $u_i\in\R$, $y_i\in\R$ are $\bb H_i$'s state, input and output respectively, and
$\vp_i : \R^{r_i} \times \R \to \R^{r_i}$ and $\varrho_i: \R^{r_i}\times \R \to \R$
are Lipschitz.
\tikzstyle{block} = [draw, rectangle, minimum height=2em, minimum width=2em, very thick]
\tikzstyle{sum} = [draw, circle]
\tikzstyle{input} = [coordinate]
\tikzstyle{output} = [coordinate]
\tikzstyle{pinstyle} = [pin edge={to-,thin,black}]
\begin{figure}[h]
\centering
\begin{tikzpicture}[auto, node distance=1.7cm,>=latex']
\node [input, name=input] {};
    \node [sum, right of=input, node distance=1cm] (sum) {};
    \node [block, right of=sum, node distance=1cm] (h1) {$\mathbb H_1$};
    \node [block, right of=h1] (h2) {$\mathbb H_2$};
    \node [right of=h2] (dots) {\ldots};
    \node [block, right of=dots] (hn) {$\mathbb H_n$};
    \node [output, right of=hn, node distance=1cm] (output) {};
\draw [draw,->] (input) -- node [near start] {$u_{ext}$} node [pos=0.9] {\tiny $+$}(sum);
\draw [->] (sum) -- node {$u_1$} (h1);
\draw [->] (h1) -- node [name=y1,near start]{$y_1$} (h2);
\draw [->] ($ (h1.east) + (5mm,0) $) -- ++(0,8mm);
\draw [->] (h2) -- node [near start] {$y_2$} (dots);
\draw [->] (dots) -- node [near start] {$y_{n-1}$} (hn);
\draw [->] (hn) -- node [name=y] {$y_n$}(output)
                -| ++(0,-1cm)
                -| node [pos=0.9] {\tiny $-$} (sum);
\end{tikzpicture}
\caption{Block diagram of a cyclic feedback system}\label{cyclic}
\end{figure}
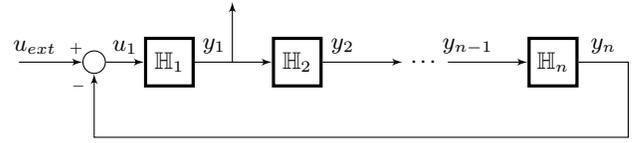
The cascaded structure of the system imposes that
\begin{eqnarray} \label{couple}
u_i & = & y_{i-1}, \qquad i=2, \ldots, n \nonumber \\
u_1 & = & u_{ext}(t)-y_{n}.
\end{eqnarray}
Therefore, the dynamics of the overall CFS can be described by
\be\label{cfs}
\begin{array}{rclrl}
\dot x_1&=&\vp_1(x_1,u_{ext}-y_n),&y_1=&\varrho_1(x_1,u_{ext}-y_n),\\
\dot x_2&=&\vp_2(x_2,y_1),        &y_2=&\varrho_2(x_2,y_1),\\
\vdots  & & & &\\
\dot x_n&=&\vp_n(x_n,y_{n-1}),    &y_n=&\varrho_n(x_n,y_{n-1}).
\end{array}
\ee

An important result of \cite[Theorem 1]{HaStSeGo12} states that a CFS, composed of iOSP blocks $\mathbb H_i$, is always iOFP with a gain, satisfying the \emph{secant} condition.
Namely, if each $\bb H_i$ is iOSP($\gamma_i^{-1},\G_{x_i}$) (with $\gamma_i>0$), then CFS \eqref{cfs} is iOFP($-k$,$\G_x$), where $\G_x\dfb\G_{x_1}\times\ldots\times \G_{x_n}$ and
\be\label{sec}
k>\bar k\dfb -\frac{1}{\gamma_1}+\gamma_2\gamma_3\ldots\gamma_n\left(cos\frac{\pi}{n}\right)^n.
\ee
Theorem~1 in \cite{HaStSeGo12} provides\footnote{Formally, this theorem deals only with the case where $\G_{x_i}=\R^{r_i}\,\forall i$, but its extension to CFS that are not globally iOSP is straightforward.} a constructive way to find the incremental storage function $V(\Delta x)$ (positive and radially unbounded), such that
\be\label{iOSPcfs}
\dot V(\Delta x)+\alpha \|\Delta y\|^2\le k(\Delta y_1)^2+\Delta y_1\Delta u_{ext}
\ee
for any two solutions staying in $\G_x$. Here $x\dfb [x_1^{\top},\ldots,x_n^{\top}]^{\top}$, $y\dfb [y_1,\ldots,y_n]^T$ stand for the joint state and output vectors respectively.

Notice that $\bar k<\infty$ if and only if all the gains $\gamma_i$ are finite, except for possibly the gain $\gamma_1$ of the leading block. For the blocks, constituting the Goodwin-type oscillator, the gain is usually finite only in a \emph{bounded} domain due to presence of the Mikhaelis-Menten or the Hill nonlinearity.\footnote{The Mikhaelis-Menten function is the nonlinear function in the form $f(x)=K_1x/(K_2+x)$, and the Hill function is the nonlinear function  in the form $f(x)=K_1/(K_2+x^p)$, where $K_1,  K_2,p>0$ are constant.} Basically, the state-space description (\ref{block}) of the block $\bb H_i$, representing a chemical reaction, is given by
\be\label{h_i}
\dot x_i=-f_i(x_i)+u_i\in\R,\, y_i=g_i(x_i)\in\R,\quad (i=1,2,\ldots,n)
\ee
where the maps $f_i,g_i$ are strictly increasing and Lipschitz continuous.

In the case when $g_i$ is linear\footnote{The linearity  follows implicitly from the assumption that the integral in \cite[eq.(8)]{HaStSeGo12} is well-defined.} ($g_i'(x)\equiv const>0$), it was proved in \cite{HaStSeGo12,Hamadeh}
that the subsystem \eqref{h_i} is iOSP($\bar \gamma_i^{-1}$), where
\be\label{gamma-wrong}
\bar \gamma_i=\sup_{x_i}\frac{g'_i(x_i)}{f_i'(x_i)}.
\ee
This claim remains valid also for nonlinear monotone functions $f_i,g_i$ if $\inf_{x_i} g_i'(x_i)>0$ \cite{Hamadeh}. The following lemma extends this result to the case where $f_i$
is not necessarily monotone and shows that the condition $\inf g_i'>0$ is critical and cannot be dropped.
\begin{lemma}\label{true-lem}
Let $\G_{x_i}\subset\R$ be an interval where the functions $f_i,g_i$ are $C^1$-smooth. For system \eqref{h_i}, the following claims hold:
\begin{enumerate}
\item If $\inf_{x_i} g'_i(x_i)>0$ and $\tilde k_i=\inf_{x_i}\frac{f'_i(x_i)}{g_i'(x_i)}>-\infty$, then \eqref{h_i} is iOFP($\tilde k_i$,$\G_{x_i}$). Hence, if $f_i$ is non-decreasing on $\G_{x_i}$, i.e. $\tilde k_i\ge 0$, then system \eqref{h_i} is iOSP($\bar\gamma_i^{-1}$,$\G_{x_i}$) with ${\bar\gamma_i}$  given in \eqref{gamma-wrong}.
    The incremental storage function $S(\Delta x)$ can be chosen quadratic.
\item Let $\G_{x_i}=(0;\infty)$, $f_i$ be globally bounded and $g_i(x)\sim a + b x^{-\alpha}$ as $x\to+\infty$, where $a,b \in \R$ and $\alpha>0$. Then \eqref{h_i} can \emph{never} be iOSP no matter what passivity gain is chosen.
\end{enumerate}
\end{lemma}
\begin{proof} We  prove 1) first. For notational simplicity, we drop the subscript $i$ in $f_i,g_i,x_i,u_i,\tilde k_i,\G_{x_i}$ throughout the proof.
By definition of $\tilde k$, one has  $f'_i(x)-\tilde k_ig_i'(x)\ge 0$ for any $x\in\G_{x_i}$.
For any pair of solutions $z^+\dfb (x^+,y^+,u^+)$ and $z^{\dagger}\dfb (x^{\dagger},y^{\dagger},u^{\dagger})$, $x^+(t),x^{\dagger}(t)\in\G_x$, let $\Delta z=(\Delta x,\Delta y,\Delta u)\dfb z^+-z^{\dagger}$, $\Delta f(t)\dfb f(x^+(t))-f(x^{\dagger}(t))$ and $\Delta g(t)\dfb g(x^+(t))-g(x^{\dagger}(t))$. Applying the mean value theorem to $f-\tilde k g$, we know that $\Delta f(t)-\tilde k\Delta g(t)=(f'(\theta(t))-\tilde k g'(\theta(t)))\Delta x(t)$ for which $\theta(t)\in\G_x$ lies between $x^+(t)$ and $x^{\dagger}(t)$, and hence $(\Delta f(t)-\tilde k \Delta g(t))\Delta x(t)\ge 0$ and $\Delta f\Delta x\ge \tilde k \Delta g\Delta x= \tilde k\Delta y\Delta x$.
Therefore, we have
\ben
\Delta u\Delta x-\tilde k\Delta y\Delta x=\Delta \dot x\Delta x+\Delta f\Delta x-\tilde k\Delta g\Delta x\ge \Delta \dot x\Delta x.
\een
Taking $S(\Delta x)\dfb \frac{1}{2\ve}|\Delta x|^2$, where $\ve:=\inf\limits_{x\in\G_{x}} g(x)>0$, one has
\ben
\Delta u\Delta y-\tilde k\Delta y^2=(\Delta u\Delta x-\tilde k\Delta y\Delta x)\frac{\Delta y}{\Delta x}>\ve\Delta\dot x\Delta x=\dot S(\Delta x),
\een
and thus \eqref{h_i} is iOSP($\tilde k$,$\G_x$), which proves statement 1).

We prove 2) by contradiction. Suppose on the contrary that \eqref{h_i} is iOSP($\tilde k,\G_{x}$) with the storage function $S(\Delta x)$. Let $M\dfb \sup\limits_{x}|f(x)|<\infty$, and   $x^{\dagger}$, $x^+$ be a pair of solutions under the inputs $u^{\dagger}(t)=M+(\gamma+1)t^{\gamma}$ and $u^+(t)=2M+\ve+u^{\dagger}(t)$ respectively, where $(\gamma+1)\alpha>1$ and $\ve>0$. Consequently,  $\dot x^{\dagger}\ge (\gamma+1)t^{\gamma}$ so that $x^{\dagger}(t)\ge t^{\gamma+1}+x^{\dagger}(0)$ as $t\to\infty$ and
$\Delta\dot x(t)\ge \ve$. Therefore, $\Delta y(t)$ is in the order of $t^{-(\gamma+1)\alpha}$ as $t\to\infty$, and thus $\Delta u\Delta y$ and $|\Delta y|^2$ are summable functions. By integrating \eqref{iOSP}, one obtains that
\ben
\begin{split}
S(\Delta x(T))-S(\Delta x(0))\le\int_0^{T}(\Delta u\Delta y-\tilde k|\Delta y|^2)dt\le\\\le\int_0^{\infty}(|\Delta u||\Delta y|+|\tilde k||\Delta y|^2)dt<\infty \textrm{\; for all } T>0,
\end{split}
\een
which contradicts the fact that $S(\Delta x(T))\to\infty$ as $T\to\infty$.
\end{proof}

\section{Synchronization of diffusively coupled CFS}\label{sec.synchro}

The iOFP property of the CFS \eqref{cfs} allows to prove synchronization in a network of $N>1$ identical CFS, where the couplings
are described by a weighted, strongly connected, balanced graph with the Laplacian matrix $L\in \R^{N\times N}$.
 Let $u_{ext}^j$, $x^j\dfb[{x_1^j}^T,\ldots,{x_n^j}^T]^T$ and $y^j\dfb[y_1^j,\ldots,y_n^j]^T$ denote respectively the external input, state and output of the $j$th CFS in the network, $j = 1, \ldots, N$. Consider the control law, forcing
the inputs of the coupled CFS in the form
\be\label{proto}
U_{ext}(t)=-L Y_1(t),
\ee
where $U_{ext} \dfb [u_{ext}^1, \ \ldots \ ,  u_{ext}^N] ^T$ and $Y_1\dfb [y_1^1,\ \ldots,\ y_1^N]^T$.

The result of \cite[Theorem 2]{HaStSeGo12} shows that protocol \eqref{proto} synchronizes the outputs of $N$ CFS's if the coupling is sufficiently ``strong''; its ``strength'' is bounded below by the algebraic connectivity of the graph if $L=L^{\top}$ or, generally, by the \emph{second smallest eigenvalue} $\la_2$ of the matrix $\frac{L + L^{\top}}{2}$. Precisely, let each CFS
satisfy \eqref{iOSPcfs} for some $\alpha>0$ and $k>0$ and be limit set detectable \cite{HaStSeGo12}. If $\la_2\ge k$, then any \emph{bounded} solution
of the coupled CFS achieves synchronization
\be\label{sync}
\lim_{t\to+\infty}|x_i^j(t)-x_i^k(t)|=0\quad\forall i=1,\ldots,n;\forall j,k=1,\ldots,N.
\ee

\begin{remark}\label{rem.thm2}
Theorem~2 in \cite{HaStSeGo12} claims a more general result stating that synchronization is achieved without the assumption on bounded solutions; however, as discussed below, this assumption is implicitly required in its proof when appealing to the LaSalle invariance principle. For special types of oscillators, e.g.
the Lur'e system with sector nonlinearity,
the solution's boundedness is ensured by the input-to-state stability property of the individual system \cite{StanSepulchre}. However, in general the technique to drop it remains elusive if not impossible.
\end{remark}

The following theorem extends Theorem~2 in \cite{HaStSeGo12} to the case where the iOFP property holds only in some domain; its proof, following the line of the proof from \cite{HaStSeGo12}, demonstrates, in particular, that the boundedness assumption is essential.
\begin{theorem}\label{thm2}
Consider a system of $N>1$ identical limit-set detectable \cite{HaStSeGo12} CFS \eqref{cfs}, satisfying \eqref{iOSPcfs} with some $\alpha,k>0$ in some \emph{closed} domain
$\G_x\subseteq\R^{r_1+\ldots+r_n}$. Suppose the CFS are coupled together through the protocol \eqref{proto} with $\la_2>k$. Then any \emph{bounded} solution of the closed-loop system, such that  $x^j(t)\in\G_x\,\forall t\ge t_0$ for $j=1,\ldots,N$ and some $t_0\ge 0$, asymptotically synchronizes \eqref{sync}.
\end{theorem}
\begin{IEEEproof}
As before we use $x^j\in \R^{r_1+\cdots+r_n}$ to denote  the state of the $j$th CFS, and now  let $\xi \dfb [{x^1}^T,\ldots,{x^N}^T]^T$ and $\zeta \dfb [{Y^1}^T,\ldots,{Y^N}^T]^T$ be the state and output of the overall networked system respectively. For the incremental storage function $V$ of each individual CFS, which satisfies \eqref{iOSPcfs}, and for $1\leq j, m\leq N$, let $V_{j,m}(\xi)\dfb V(x^m-x^j)$ and $S(\xi)\dfb\frac{1}{2N}\sum_{j,m=1}^NV_{j,m}(\xi)$. Substituting solutions $(x^p,u^p,y^p)$ and $(x^q,u^q,y^q)$, where $x^p(t),x^q(t)\in\G_x$ for $t\ge t_0$,
into \eqref{iOSPcfs}, the following condition is valid as $t\ge t_0$:
\ben
\dot {V}(x^p-x^q)+\alpha \|y^p-y^q\|^2\le k(y_1^p-y_1^q)^2+(y_1^p-y_1^q)(u_{ext}^p-u_{ext}^q).
\een
By summing up these inequalities over all $p,q$ and introducing the projector $\Pi\dfb I_N-\frac{1}{N}\mathbf{1}_N\mathbf{1}_N^T$, one arrives at \cite{HaStSeGo12} the following
\be\label{aux0}
\dot S(\xi)\le -\alpha\left\|(\Pi\otimes I_n)\zeta\right\|^2+\left(k\|\Pi Y_1\|^2+(\Pi Y_1)^T\Pi U_{ext}\right).
\ee
Using \eqref{proto}, one easily finds that $\Pi U_{ext}=-\Pi LY_1=-L\Pi Y_1$ and hence $(\Pi Y_1)^T\Pi U_{ext}\le -\la_2\|\Pi Y_1\|^2$; therefore \eqref{aux0} implies that
\be\label{aux00}
\dot S(\xi)\le -\alpha\left\|(\Pi\otimes I_n)\zeta(t)\right\|^2\le 0\quad\forall t\ge t_0.
\ee

Let $\xi(t)$ be a bounded solution with $x^j(t)\in\G_x$ and $\scr M\dfb\{\xi=(x^1,\ldots,x^N)\in\G_x^N: S(\xi)\le S(\xi(t_0))\}$.
Due to \eqref{aux0} one has $\xi(t)\in\scr M$ and hence the closed set $\scr M$ contains the $\om$-limit set of $\xi(\cdot)$.
Thanks to the LaSalle invariance principle, the solution $\xi(t)$ converges to the maximal subset of $\scr M$, where $\dot S=0$ and
hence $y^1=\cdots=y^N$ due to \eqref{aux00}. The limit-set detectability assumption entails now synchronization of the state vectors.
\end{IEEEproof}

Note that a widely used version of LaSalle invariance principle \cite{Kh96} requires the Lyapunov function $S$, along with $\dot S$, to be defined on a \emph{compact invariant} set $\scr M$, and guarantees that any solution starting in $\scr M$ converges to the maximal set where $\dot S=0$. However, original versions of LaSalle's invariance principle \cite{La60,La76} are applicable to any \emph{bounded solution} and guarantee that $\dot S\equiv 0$ on its $\omega$-limit set, provided that $S$ and $\dot S$ are well defined in the vicinity.
Assumption of compactness and invariance of $\scr M$ automatically entail the latter condition, as well as boundedness of any solution starting at $\scr M$. Without this assumption, LaSalle's invariance principle can still be applied, but the extra condition of \emph{boundedness} is then unavoidable.

To prove synchronization of CFS under linear balanced protocol \eqref{proto}, using Theorem~2 in \cite{HaStSeGo12} or more general Theorem~\ref{thm2}, one has first to establish the iOFP property in some domain $\G_x$. As follows from Lemma~\ref{true-lem}, the relevant passivity gains $\gamma_i$ of the subsystems $\bb H_i$ can be infinite or even undefined, unless the corresponding state
variables $x_i$ are confined to some bounded domains $\G_{x_i}$. So the restriction $x(t)\in G_x$, imposed to apply the iOFP property \eqref{iOSPcfs}, requires to find some \emph{explicit} bound for the solution. Even if $\G_x$ can be unbounded (like in the example from \cite{HaStSeGo12}), the criterion still guarantees synchronization only for \emph{bounded} solutions.
Using \eqref{aux00}, deviations $x^p-x^q$ are shown to be bounded, entailing boundedness of the states $x^j(t)$ under input-to-state stability assumptions
(which hold e.g. for Lur'e-type systems \cite{StanSepulchre}). However, proving the solution boundedness for general CFS, coupled via a linear protocol \eqref{proto}, remains a non-trivial problem.

To cope with this problem, we replace the linear protocol \eqref{proto} with a \emph{nonlinear} one, providing sufficiently small and non-negative control inputs $u_{ext}^j$, that is,
$0\le u_{ext}^j(t)\le M_0$, where $M_0$ is some known constant. Under such a constraint, one often can localize the solution in a domain where the incremental passivity gains of all the subsystems are known and finite. Relevant sufficient conditions for this, dealing with Goodwin's oscillators, will be discussed in Section~\ref{sec.appl}.
In fact, the input restrictions are often dictated by the biological feasibility, e.g. in some models of coupled \emph{circadian clocks} \cite{GoBeWaKrHe05,LoWeKrHe08,VaHe11} the oscillators' inputs stand for the concentrations of the neurotransmitter in extracellular media. In this technical note, we do not aim to examine the model from \cite{GoBeWaKrHe05,LoWeKrHe08,VaHe11} itself, which considers \emph{mean field} couplings. Instead, we propose a \emph{diffusive} coupling protocol similar to \eqref{proto}, but employing a non-negative ``saturating'' nonlinearity, which guarantees
the input constraint and thus entails the solution's boundedness. Meanwhile, the protocol constructed below provides synchronization under sufficiently ``strong'' coupling, and the minimal sufficient strength may also be explicitly estimated.

The algorithm we propose is as follows
\be\label{proto1}
u_{ext}^j(t)=g_0(cv^j(t)),\;\; V(t)\dfb [v^1,\ldots,v^N]^{\top}=-L Y_1(t).
\ee
A constant $c$ stands for the coupling gain; here, to emphasize the effect of the coupling strength, we have intentionally added $c$ that has been implicitly incorporated into the entries of $L$ in \eqref{proto} as is done in \cite{HaStSeGo12}.
The function $g_0:\R\to [0;+\infty)$ is bounded, saturating the inputs at a prescribed constant $M_0=\sup_{v\in\r} g_0(v)$. The auxiliary inputs $v^j(t)\in\R$ are introduced to emphasize the similarity between the protocols \eqref{proto} and \eqref{proto1}: in fact, the system of $N$ CFS's \eqref{cfs}, coupled through the protocol \eqref{proto1}, may be considered as a collective of appropriately modified CFS's, coupled linearly.

Hereinafter, we assume the following assumption to be valid.
\begin{assumption}\label{assum.g0}
The function $g_0(\cdot)$ is smooth, globally bounded and strictly increasing, hence $g_0'(v)>0,\,\forall v\in\R$.
Additionally, $\nu(s)\dfb \inf\limits_{|v|\le s}g_0'(v)$ decreases at infinity more slowly than linear functions, i.e. $\nu(s)\to 0$ yet $|s|\nu(s)\to+\infty$ as $s\to\pm\infty$.
\end{assumption}

Assumption~\ref{assum.g0} is satisfied by a wide class of functions, e.g.
\be\label{g0}
g_0(v)=\frac{M_0}{2}\left (1+\frac{|v|^{\rho}\textrm{sign}\, v}{1+|v|^{\rho}}\right ),\quad \textrm{with }M_0>0,\rho\in (0,1).
\ee
Note that $M_0=\sup\limits_{v\in\r} g_0(v)$ may be chosen as small as possible.

The next result shows that the modified protocol \eqref{proto1} synchronizes the systems \eqref{cfs}, provided that the coupling is sufficiently strong and the solution stays in some compact set; the crucial difference with the linear protocol \eqref{proto} is that the
existence of such a set attracting the solutions may often be proved by choosing $g_0(\cdot)$ sufficiently small.
\begin{theorem}\label{thm3}
Suppose that Assumption~\ref{assum.g0}, and the assumptions of Theorem~\ref{thm2} hold, where $\G_x\subset\R^{r_1+\ldots+r_n}$ is a \emph{compact} set and $\rho_1(x_1,u_1)=\rho_1(x_1)$.
Then for sufficiently large gain $c$, any solution of the closed-loop system such that $x^p(t)\in\G_x\,\forall t\ge t_0$ for all $p=1,\ldots,N$ and some $t_0\ge 0$ gets synchronized \eqref{sync}. Synchronization is implied by the following inequality, which holds as $c\to+\infty$
\begin{equation}\label{eq.strong}
c\nu(c\|L\|\bar y_*)\la_2>k,\; y_*\dfb\max\{|\rho_1(x_1)|:x\in\G_x\}.
\end{equation}
\end{theorem}
\begin{IEEEproof}
Along with the original CFS \eqref{cfs}, consider a modified system with a new input $v(t)$, which obeys \eqref{cfs} and the additional equation $u_{ext}=cg_0(cv)$ as shown in Fig.~\ref{cyclic0}.
The network of CFS \eqref{cfs}, coupled via protocol \eqref{proto1}, is now equivalent to the network of $N$ ``augmented'' systems, coupled via \eqref{proto}.
One may easily notice that if $x^p(t)\in\G_x$ then $|y_1^p(t)|\le y_*$ and hence $|v^p(t)|\le \|L\|_{\infty}y_*$ due to \eqref{proto};
this implies that $0\le u_{ext}^p(t)\le u_*\dfb\max g_0$. Due to mean value theorem, for any two solutions one has $\Delta u_{ext}(t)=cg_0'(c\theta(t))\Delta v$, where $|\theta(t)|\le \|L\|y_*$, and thus $0<m\dfb c\nu(c\|L\|_{\infty}y_*)\le cg_0'(c\theta)\le M\dfb c\max\limits_{v\in\R} g_0'(v)$ and thus the right-hand side of \eqref{iOSPcfs} is not greater than $M(m^{-1}k|\Delta y_1(t)|^2+\Delta v(t)\Delta y_1(t))$. This ensures that the augmented CFS also satisfies \eqref{iOSPcfs}, where $u_{ext}$, $\alpha$ and $k$ are to be replaced with respectively $v$, $\tilde\alpha\dfb M^{-1}\alpha$ and $\tilde k=k/m$.
Synchronization now follows\footnote{Formally, Theorem~\ref{thm2} was formulated for systems \eqref{cfs}. However, its proof employs only the inequality \eqref{iOSPcfs} and obviously remains valid in spite of the additional static block $\bb H_0$.} from Theorem~\ref{thm2} since $\la_2>\tilde k$ due to \eqref{eq.strong}.
\end{IEEEproof}
\begin{figure}
\centering
\begin{tikzpicture}[auto, node distance=1.7cm,>=latex']
\node [input, name=input] {};
    \node [block, right of = input, node distance=1cm] (h0) {$\mathbb H_0$};
    \node [sum, right of=h0, node distance=1.3cm] (sum) {};
    \node [block, right of=sum, node distance=1cm] (h1) {$\mathbb H_1$};
    \node [block, right of=h1, node distance=1.4cm] (h2) {$\mathbb H_2$};
    \node [right of=h2, node distance=1.4cm] (dots) {\ldots};
    \node [block, right of=dots, node distance=1.4cm] (hn) {$\mathbb H_n$};
    \node [output, right of=hn, node distance=1cm] (output) {};

\draw [->] (input) -- node [near start] {$v$} (h0);
\draw [->] (h0) -- node [near start] {$\,\,\,u_{ext}$} node [pos=0.9] {\tiny $+$}(sum);
\draw [->] (sum) -- node {$u_1$} (h1);
\draw [->] (h1) -- node [name=y1,near start]{$y_1$} (h2);
\draw [->] ($ (h1.east) + (4mm,0) $) -- ++(0,8mm);
\draw [->] (h2) -- node [near start] {$y_2$} (dots);
\draw [->] (dots) -- node [near start] {$y_{n-1}$} (hn);
\draw [->] (hn) -- node [name=y] {$y_n$}(output)
                -| ++(0,-1cm)
                -| node [pos=0.9] {\tiny $-$} (sum);

\end{tikzpicture}
\caption{Auxiliary cyclic feedback system with saturated input}\label{cyclic0}
\end{figure}
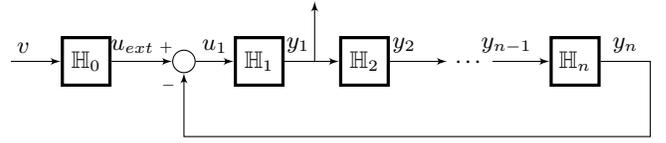

\section{Applications in coupled biochemical oscillators}\label{sec.appl}

In this section, we discuss how Theorem~\ref{thm3} allows to estimate the coupling strength, needed to synchronize biochemical
oscillators of the Goodwin type, governed by the equations
\be\label{cfs1}
\begin{array}{rclrcl}
\dot x_1&=&-f_1(x_1)+(u_{ext}-y_n),\;   y_1=g_1(x_1),\\
\dot x_2&=&-f_2(x_2)+ y_1,\qquad\qquad\;        y_2=g_2(x_2),\\
\vdots &  & \\
\dot x_{n-1}&=&-f_{n-1}(x_{n-1})+ y_{n-2},\; y_{n-1}=g_{n-1}(x_{n-1}),\\
y_n&=&g_n(y_{n-1}).
\end{array}
\ee
Therefore, the Goodwin-type oscillator is an example of the system \eqref{cfs} shown in Fig.~\ref{cyclic}, where the blocks $\bb H_1,\ldots, \bb H_{n-1}$ obey equations \eqref{h_i} with $u_1 =  u_{ext}-y_n$, $u_2= y_1$, \ldots, $u_{n}= y_{n-1}$, and the block $\bb H_n$ is static: $y_{n}(t)=g_n(u_n(t))$. The classical Goodwin's model \cite{Goodwin65} corresponds to the case where
$f_i$ are \emph{linear} and $g_n$ is the Hill nonlinearity. We emphasize that the system operates in the positive orthant, namely $x_i\ge 0$, $1\leq i \leq n-1$.

Consider now a network of $N$ identical oscillators \eqref{cfs1}, coupled via the distributed protocol \eqref{proto1}, where $g_0$ satisfies Assumption~\ref{assum.g0}.

Denoting $M_0\dfb \max\limits_{v\in\r} g_0(v)$, we adopt the following assumption.
\begin{assumption} \label{assumption}
The functions $f_i,g_i:\overline\R_+\mapsto \overline\R_+$, $i=1,\ldots,n-1$, $g_0:\R\mapsto [0,M_0], g_n:\overline \R_+\mapsto [-M_n,0]$ are smooth and strictly increasing ($f_i',g_i'>0$).
Additionally, the maps $f_i$ satisfy the condition
\be\label{Phi}
\Phi_i(0)=-\infty,\Phi_i(+\infty)=+\infty,\,\, \Phi_i(x)\dfb \int_1^x\frac{ds}{f_i(s)}, x\ge 0.
\ee
\end{assumption}


Condition \eqref{Phi} holds, for instance, for linear functions $f_i(x)=a_i x$, $a_i>0$ and Mikhaelis-Menten type nonlinear functions.
Under Assumption \ref{assumption}, any solution of the cyclic feedback system with saturated input, starting strictly inside the positive orthant, remains positive and, under additional assumptions,
is ultimately bounded.
\begin{lemma}\label{lem.tech}
   Let Assumption~\ref{assumption} hold and $u_{ext}(t)\in [0;M_0]$ is defined for $t\in\Delta\dfb [0;\beta)$. Then for any initial condition $x_i(0)\in\R_+$ ($i= 1, \ldots, n-1$) the
solution $x_1(t),\ldots,x_{n-1}(t)$, $y_1(t),\ldots,y_n(t)$ exists on $\Delta$ and remains positive $x_i(t)\in \R_+\,\forall t\in\Delta.$
If $\beta=\infty$ and the functions $h_i$, given by the recursion $h_1=f_1^{-1}$, $h_2=f_2^{-1}\circ g_1\circ h_1$,\ldots, $h_{n-1}=f_{n-1}^{-1}\circ g_{n-2}\circ h_{n-2}$, are well defined on
$[0;M+\ve_0)$, where $M\dfb M_0+M_n$ and $\ve_0>0$, then the solution is ultimately bounded
\be\label{limsup}
\varlimsup\limits_{t\to+\infty} x_i(t)\le \bar x_i\,\quad \forall i=1,\ldots,n-1\, \forall x_i(0),
\ee
where $\bar x_i\dfb h_i(M)$ are \emph{independent} of the initial conditions.
\end{lemma}
\begin{IEEEproof}
We now prove that $x(t)$ exists and is positive on $\Delta$. Since $x_i(0)>0$, because of continuity,  there exists a maximal interval $\Delta'=[0;\beta')\subseteq\Delta$, such that $x_i(t)>0$ for all $t\in \Delta'$ and $i=1,\ldots,n-1$. From Assumption~\ref{assumption}, for $t\in\Delta'$ it always holds that $u_i(t)\ge 0$
since $y_n(t)\le 0$, $u_1 =  u_{ext}-y_n$, $u_2= y_1$, \ldots, $u_{n}= y_{n-1}$. Therefore
$\frac{d}{dt}\Phi_i(x_i(t))=\dot x_i(t)/f_i(x_i(t))\ge -1$, which implies that $x_i(t)>\Phi_i^{-1}(\Phi_i(0)-t)>0$ for
any $t\in\Delta'$. Hence, the solution cannot escape from $\R_+$ within $\Delta'$. Furthermore, if $\beta'<\infty$ then $x_i(t)>\delta_i>0$ and hence $f_i(x_i(t))\ge\upsilon_i>0$ as $t\in\Delta'$.
This implies that $\Phi_i(x_i(t))$ and hence $x_i(t)$ are bounded from above, i.e. the solution cannot grow unbounded in finite time.
This, according to the definition of $\Delta'$, implies that $\Delta'=\Delta$.

To prove \eqref{limsup}, we show first that for any solution $x_i(t)>0,u_i(t)>0$ of the subsystem \eqref{h_i} the following implication holds:
 \be\label{limsup1}
\mu_i\dfb \varlimsup\limits_{t\to+\infty} u_i(t)<\sup\limits_{x>0}f_i(x)\Longrightarrow
\varlimsup\limits_{t\to\infty}x_i(t)\le f_i^{-1}(\mu_i).
\ee
Indeed, let $\delta,\ve>0$ be so small that $\mu_i+\delta+\ve<\sup f_i$. From the definition of the upper limit, a number $T_0$ exists such that
$u_i(t)\le \mu_i+\delta$ for all $t\ge T_0$. From the facts that $\dot x_i(t)<-\ve$ if $t\ge T_0$ and $x_i(t)>\xi\dfb f_i^{-1}(\mu_i+\delta+\ve)\Leftrightarrow f_i(x_i(t))>\mu_i+\delta+\ve\ge u_i(t)+\ve$, we know the following two statements hold: (i) if $x_i(T_1)\le\xi$ for some $T_1\ge T_0$, then $x_i(t)\le\xi\,\forall t\ge T_1$, and (ii) such a $T_1$ necessarily exists; that is, if $x_i(T_0)\le \xi$, one can take $T_1\dfb T_0$, otherwise, $T_1\le (x_i(T_0)-\xi)/\ve$ is the first time instant after $T_0$ at which $x_i(T_1)=\xi$. Hence, $\varlimsup\limits_{t\to+\infty} x_i(t)\le f_i^{-1}(\mu_i+\delta+\ve)$, from which \eqref{limsup1} follows by passing to the limit $\delta,\ve\to 0$.

Then \eqref{limsup} can be easily proved using \eqref{limsup1}. From the assumptions, the input of the block $\bb H_1$ is given by $u_1(t)=u_{ext}(t)-y_n(t)\le M$. Therefore, $\varlimsup\limits_{t\to\infty} x_1(t)\le f_1^{-1}(M)=h_1(M)$ and hence $\varlimsup\limits_{t\to\infty} y_1(t)\le g_1\circ h_1(M)$.
Invoking \eqref{limsup1} for the second block $\bb H_2$ with input $u_2=y_1$, we obtain $\varlimsup\limits_{t\to\infty} x_2(t)\le f_2^{-1}\circ g_1\circ h_1(M)=h_2(M)$. Iterating this procedure for $\bb H_3,\ldots, \bb H_{n-1}$, the inequalities \eqref{limsup} are proved.
\end{IEEEproof}

Lemma~\ref{lem.tech} gives only the simplest condition of ``restricted'' input-to-state stability (ISS) \cite{SontagWang95}, that is, the existence of explicit ultimate bounds for the state vector of CFS provided that its input is sufficiently small and positive. This condition appears to be conservative for some Goodwin-type oscillators, as will be discussed below.
It can be further refined (with tightening the bounds $\bar x^i$) by using the monotonicity-type arguments from \cite{Al77},\cite{EncisoSmithSontag06}. To establish the ISS property for general CFS
(with explicit bounds $\bar x^i$) remains an open non-trivial problem, which is beyond the scope of this technical note. However, the following simple lemma shows that in practice \eqref{limsup}
holds as $u_{ext}(t)$ is sufficiently small, provided that the oscillators have a globally stable attractor (e.g. limit cycle).
\begin{lemma}\label{lem.tech2}
Suppose that any solution of the system \eqref{cfs1} with $u_{ext}\equiv 0$, starting at a compact $K\subset\R_+^{n-1}$, converges to some attractor $K_0\subseteq Int\,K$: $dist(x(t),K_0)\to 0$ as $t\to\infty$. Then for any $\delta>0$ there exist $\ve_0=\ve_0(K_0,K)>0$ such that any solution of \eqref{cfs1}, starting at $x(0)\in K$ and associated with input $0\le u_{ext}(t)\le\ve_0$, converges to the attractor's $\delta$-neighborhood. Precisely, there exists $T_0=T_0(\delta,K,K_0)$ such that $dist(x(t),K_0)<\delta$ as $t\ge T_0$.
\end{lemma}
\begin{IEEEproof}
Without loss of generality, let $\delta>0$ be so small that $K_{\delta}=\{x:dist(x(t),K_0)<\delta\}\subset K$. Since $K$ is compact, there exists $T_0$ such that $x(t)\in K_{\delta}$ as $t\ge T_0$ under $u_{ext}\equiv 0$ for any solution, starting at $x(0)\in K$. Hence, if $0\le u_{ext}(t)\le \ve_0\,\forall t\in [0;2T_0]$ and $\ve_0>0$ is sufficiently small, one can guarantee that $x(t)\in K_{\delta}\subset K$ at least for $t\in [T_0;2T_0]$ \emph{independent} of the initial condition in $K$ and concrete $u_{ext}$. Applying this for $x(0)=x(T_0)\in K$ and shifted input
$\tilde u_{ext}(s)=u_{ext}(T_0+s)$, one shows that $x(t)\in K_{\delta}\subset K$ as $t\in [2T_0;3T_0]$, and so on.
\end{IEEEproof}

Assumptions of Lemma~\ref{lem.tech2} in general hold for biologically realistic models, where the oscillating concentrations of the reagents are confined to some
(roughly known) intervals and the limit cycles are found experimentally or via numerical simulations. Lemma~\ref{lem.tech2} states that knowledge of the attractor allows to estimate the solutions of the network of coupled oscillators, provided that the control inputs are sufficiently small.
However, unlike Lemma~\ref{lem.tech}, Lemma~\ref{lem.tech2} does not give the explicit dependence between $\bar x^i$ and the value of $M_0\max u_{ext}$, but only allows to find the limit of
$\bar x^i$ as $M_0\to 0$ (and $x(0)\in K$).

We now return to the dynamics of coupled CFS \eqref{cfs1} under ``saturated'' protocol~\eqref{proto1} and prove that if the ultimate boundedness \eqref{limsup} holds under ``weak'' non-negative inputs $u_{ext}(t)$ (for instance, condition from Lemma~\ref{lem.tech} or Lemma~\ref{lem.tech2} is valid), then \eqref{proto1} guarantees synchronization for small $M_0$ is small, if the coupling is sufficiently strong: $c>c_*=c_*(\bar x_i)$ (and $c_*$ can be found explicitly).

For convenience, we introduce the following assumption.
\begin{assumption}\label{ass.ultimate}
For some set $\G\subset\r_+^{n-1}$ there exist such $M_0>0$ such that under any input $u_{ext}(t)\in [0;M_0]\,\forall t\ge 0$
solutions of \eqref{cfs1}, starting at $x(0)\in \G$, satisfy \eqref{limsup}, where the bounds $\bar x^i$ are uniform over all $x(0)$ and $u_{ext}(\cdot)$.
\end{assumption}

Assumption~\ref{ass.ultimate} can be provided, for instance, by the conditions from Lemma~\ref{lem.tech} or Lemma~\ref{lem.tech2}. If it holds, any solution which starts at $\G$
enters in finite time the hypercube $\scr B_{\ve}=\{x: 0 < x_i < \bar x_i^{\ve}\dfb \bar x_i+\ve,\, 1\le i<n\}$ and remains there.
Since the blocks $\bb H_i$ are iOSP($(\gamma_i^{\ve})^{-1},[0;\bar x_i^{\ve}]$) for $i=1,\ldots,n$ by virtue of Lemma~\ref{true-lem}, where
the secant gains are given respectively by
\be\label{gamma-e}
\gamma_i^{\ve}\dfb \max_{x_i\in [0,\bar x_i^{\ve}]}\frac{g_i'(x_i)}{f_i'(x_i)},\,i<n,\;\gamma_n^{\ve}\dfb \max_{y\in [0,g_{n-1}(\bar x_i^{\ve})]} g_n'(y).
\ee
As follows from Theorem~1 in \cite{HaStSeGo12} (see discussion in Section~II-B), the CFS \eqref{cfs1} is iOFP(-$k_{\ve}$,$\scr B_{\ve}$) and, moreover, the inequality holds
\be\label{aux1}
\dot {V}_{\ve}(\Delta x(t))+\alpha_{\ve} \|\Delta y(t)\|^2\le k_{\ve}(\Delta y_1(t))^2+\Delta y_1(t)\Delta u_{ext}(t),
\ee
for any two solutions of \eqref{cfs1}, staying in $\scr B_{\ve}$. Here $k_{\ve}=-\frac{1}{\gamma_1^{\ve}}+\gamma_2^{\ve}\gamma_3^{\ve}\ldots\gamma_{n}^{\ve}$. Notice that $\ve>0$ can be arbitrarily small, \emph{independent} of the initial condition. If $x(t)\in \scr B_{\ve}$, then $0<y_1(t)\le \bar y_*^{\ve}\dfb g_1(x_1^{\ve})$.

\begin{theorem}\label{thm-my}
Let $N$ identical CFS \eqref{cfs1} satisfy Assumptions~\ref{assumption} and ~\ref{ass.ultimate} and be coupled via \eqref{proto1}.
Suppose that $g_0$ satisfies Assumption~\ref{assum.g0}, $g_0(v)\le M_0\,\forall v\in\r$ and the coupling is sufficiently strong
\be\label{synchronize}
c\nu(c\|L\|\bar y_*^0)\la_2>-\frac{1}{\gamma_1^0}+\gamma_2^0\gamma_3^0\ldots\gamma_{n}^0,\quad \gamma_i^0\dfb \gamma_i^{\ve}|_{\ve=0},y_*^0\dfb y_*^{\ve}|_{\ve=0}.
\ee
Under such coupling, any solutions $(x^j(t))_{j=1}^N$, starting at $x^p(0)\in \G$, remain in the positive orthant $x_i^j(t)>0$ and synchronizes \eqref{sync}.
\end{theorem}
\begin{IEEEproof}
Since $\gamma_i^{\ve}$ is continuous at $\ve=0$, inequality \eqref{synchronize} remains valid after replacing
$\gamma_i^0\mapsto \gamma_i^{\ve}$ and $\bar y_1^0\mapsto \bar y_1^{\ve}$ when $\ve>0$ is small.
The claim now follows from \eqref{aux1} and Theorem~\ref{thm3}.
\end{IEEEproof}

\section{Example: synchronization of circadian clocks}

In this section we demonstrate our synchronization criterion for oscillators, describing the cellular circadian clocks \cite{GoBeWaKrHe05,LoWeKrHe08}. The main circadian pacemaker in mammal
is controlled by the neurons of SCN (suprachiasmatic nucleus, a zone in hypothalamus). Within each cell (indexed $1$ through $N$), a clock gene mRNA (X) produces a clock protein (Y) which, in turn, activates a transcriptional inhibitor (Z), closing a negative feedback loop \cite{LoWeKrHe08}; their dynamics are given by
\be\label{gonze}
\begin{split}
\dot X^i(t)&=\nu_1\frac{K_1^n}{K_1^n+(Z^i)^n}-\nu_2\frac{X^i}{K_2+X^i}+u_{ext}^i\\
\dot Y^i(t)&=k_3X^i-\nu_4\frac{Y^i}{K_4+Y^i}\\
\dot Z^i(t)&=k_5Y^i-\nu_6\frac{Z^i}{K_6+Z^i},\quad i=1,2,\ldots,N\\
\end{split}
\ee
Here $u_{ext}^i$ are some external inputs. In \cite{GoBeWaKrHe05,LoWeKrHe08,VaHe11} the networks with \emph{mean-field couplings} are considered, where $u_{ext}^1=\ldots=u_{ext}^N$ is a common input, which is positive, bounded and depend on the average concentration of neurotransmitting peptide in the extracellular domain, depending in its turn on $X^1,\ldots,X^N$.
We consider synchronization of oscillators \eqref{gonze} under \emph{diffusive} protocol \eqref{proto1}. Unlike the mean-field control, which remains oscillatory when the synchronization is established, under protocol \eqref{proto1} the inputs stabilize at the constant value $u_{ext}^i(t)\to const=g(0)\,\forall i$ as $t\to\infty$.

It is confirmed experimentally \cite{GoBeWaKrHe05} that an individual circadian clock has a stable limit cycle in the positive octant $X,Y,Z>0$,
and the corresponding oscillation period lies between 20 and 27 hours. This means that for realistic sets of parameters in \eqref{gonze} Lemma~\ref{lem.tech2} and Theorem~\ref{thm-my}
work, stating that for any compact set $K\subset\R_+^3$ one can find $M_0=M_0(K)>0$ (sufficiently small) and $c>0$ (sufficiently large), such that the protocol \eqref{proto1} synchronizes oscillators \eqref{gonze} starting at $(X^i(0),Y^i(0),Z^i(0))\in K$. For a special set of parameters $n,k_i,K_i,\nu_i$, found in \cite{LoWeKrHe08}, Lemma~\ref{lem.tech} is applicable which gives explicit estimates for $M_0$ and explicit bounds for the solutions. We simulated dynamics of a more complicated model \cite{GoBeWaKrHe05}, where Lemma~\ref{lem.tech} is unapplicable (in the notation of Lemma~\ref{lem.tech}, the function $h_1=f_1^{-1}$ is not defined on $[0;M_n]$) and $M_0$, $c$ are to be found numerically.
Note that the incremental passivity of \eqref{gonze} in the whole positive orthant does not follow from the criterion in \cite{HaStSeGo12}, unlike the Goodwin oscillator in the example from \cite{HaStSeGo12}, so synchronization of CFS \eqref{gonze} under strong \emph{linear} couplings remains an open problem.

 We simulate the dynamics of $N=10$ all-to-all coupled oscillators \eqref{gonze} with the parameters from \cite{GoBeWaKrHe05}: $\nu_1=0.7 nM/h$; $K_1=1nM$; $n=4$; $\nu_2=0.35nM/h$; $K_2=nM$; $k_3=0.7/h$; $\nu_4=0.35 nM/h$; $K_4=1nM$; $k_5=0.7/h$; $\nu_6=0.35nM/h$; $K_6=1 nM$, which correspond to the oscillation period $\approx 23.5h$.
We choose $g_0(v)$ in the form \eqref{g0}, where $M_0=0.0005$ and $\rho=0.9$. We simulate the dynamics for $c=0$, $c=1$, $c=10$ and $c=100$. Oscillators are not synchronous for $c$ being small, however, the synchronization emerges as $c$ increases, confirming thus Theorem~\ref{thm-my}.

\begin{figure}
\begin{subfigure}[b]{0.5\columnwidth}
\includegraphics[width = \columnwidth]{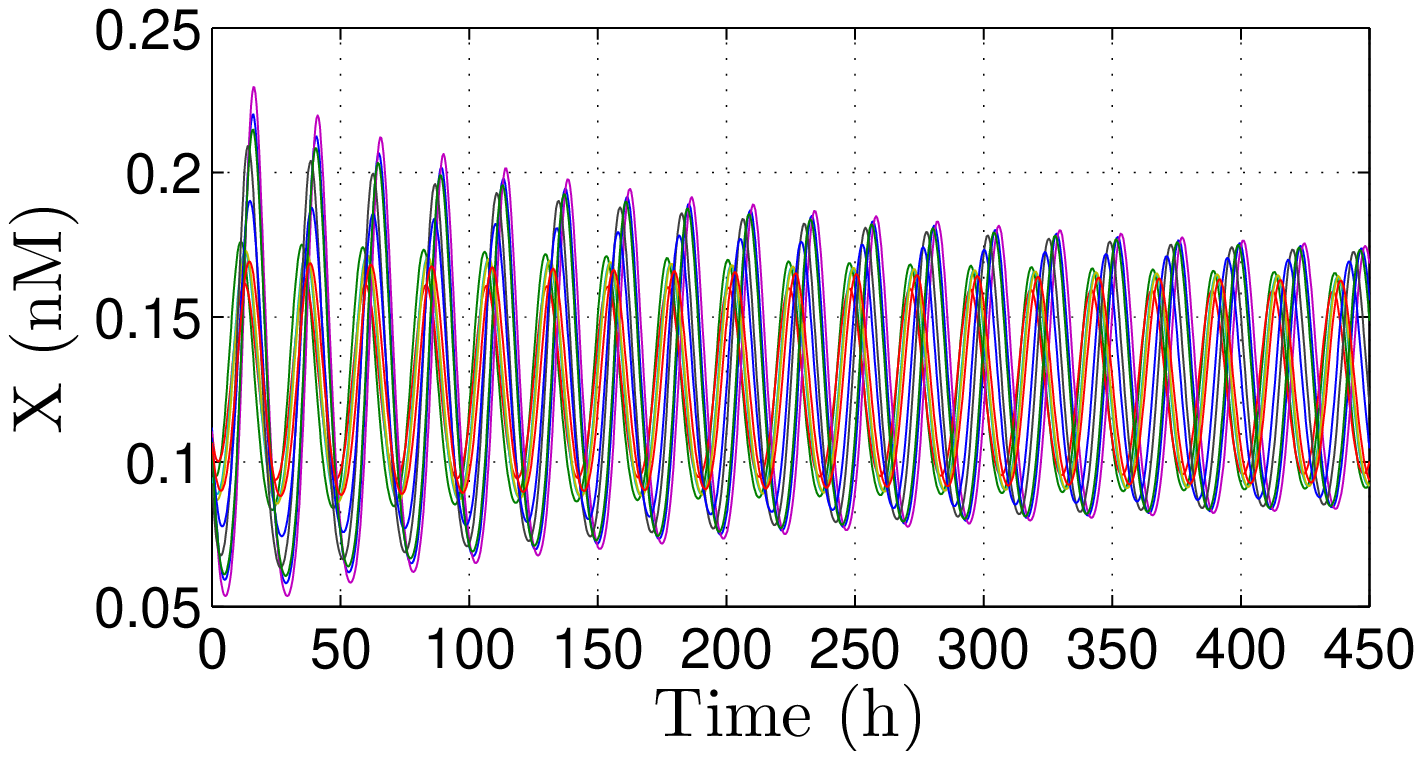}
\caption*{$c=0$}
\end{subfigure}
\begin{subfigure}[b]{0.5\columnwidth}
\includegraphics[width = \columnwidth]{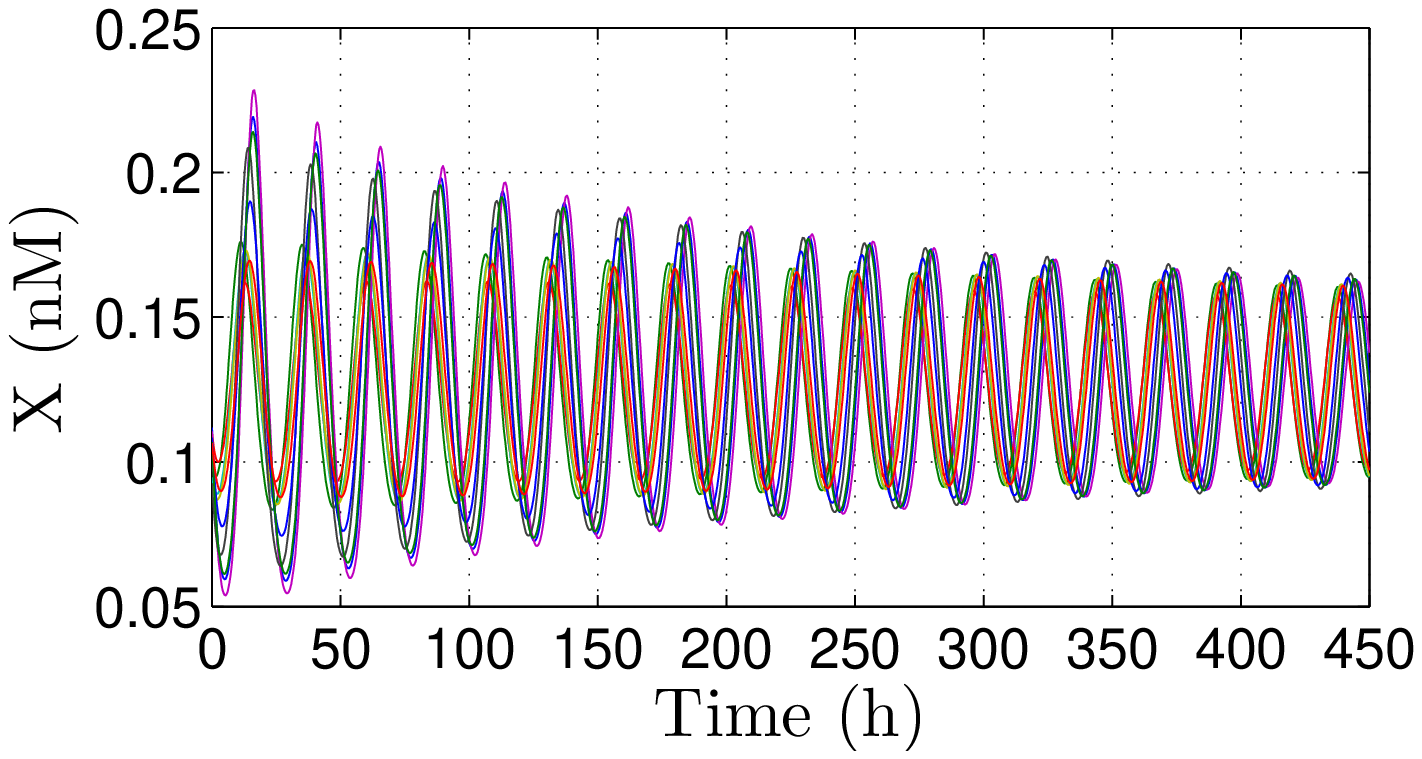}
\caption*{$c=1$}
\end{subfigure}
\begin{subfigure}[b]{0.5\columnwidth}
\includegraphics[width = \columnwidth]{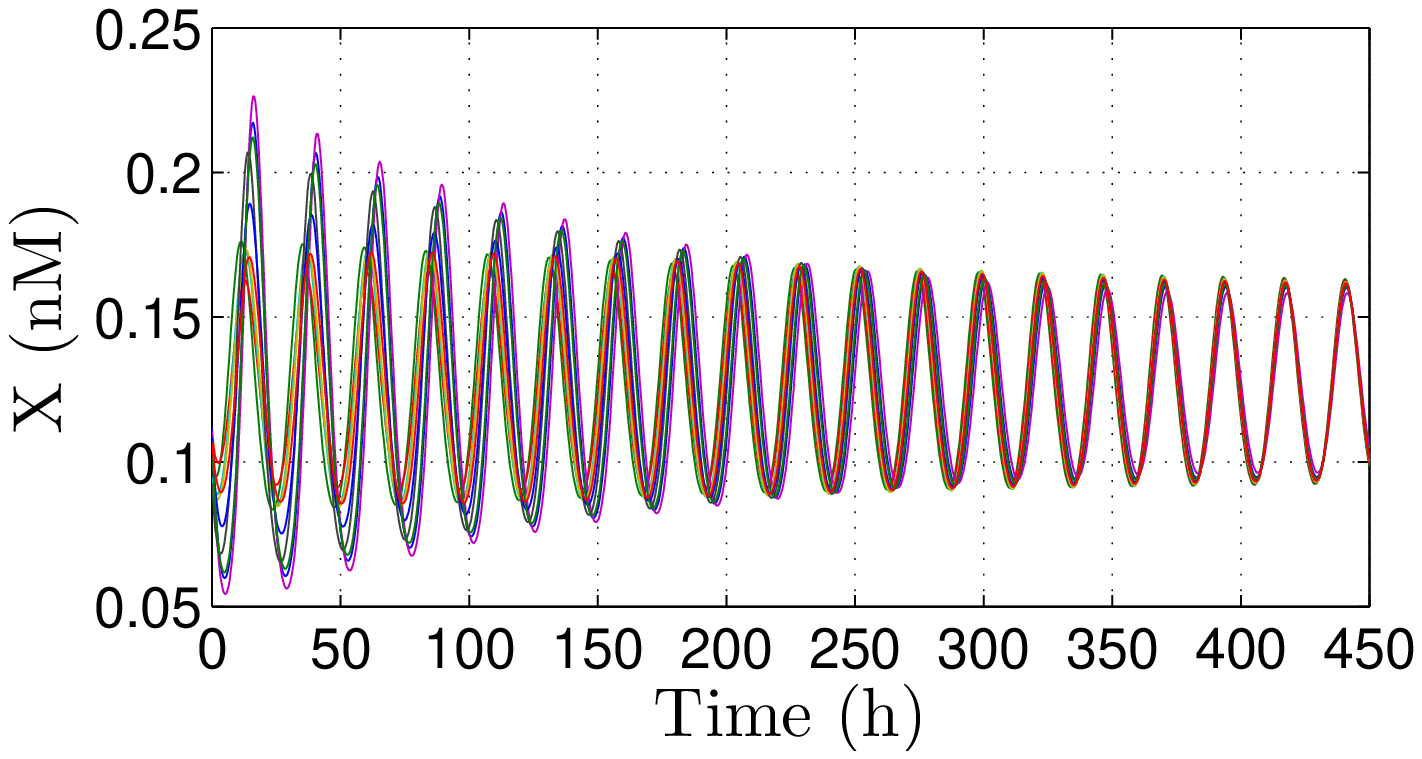}
\caption*{$c=10$}
\end{subfigure}
\begin{subfigure}[b]{0.5\columnwidth}
\includegraphics[width = \columnwidth]{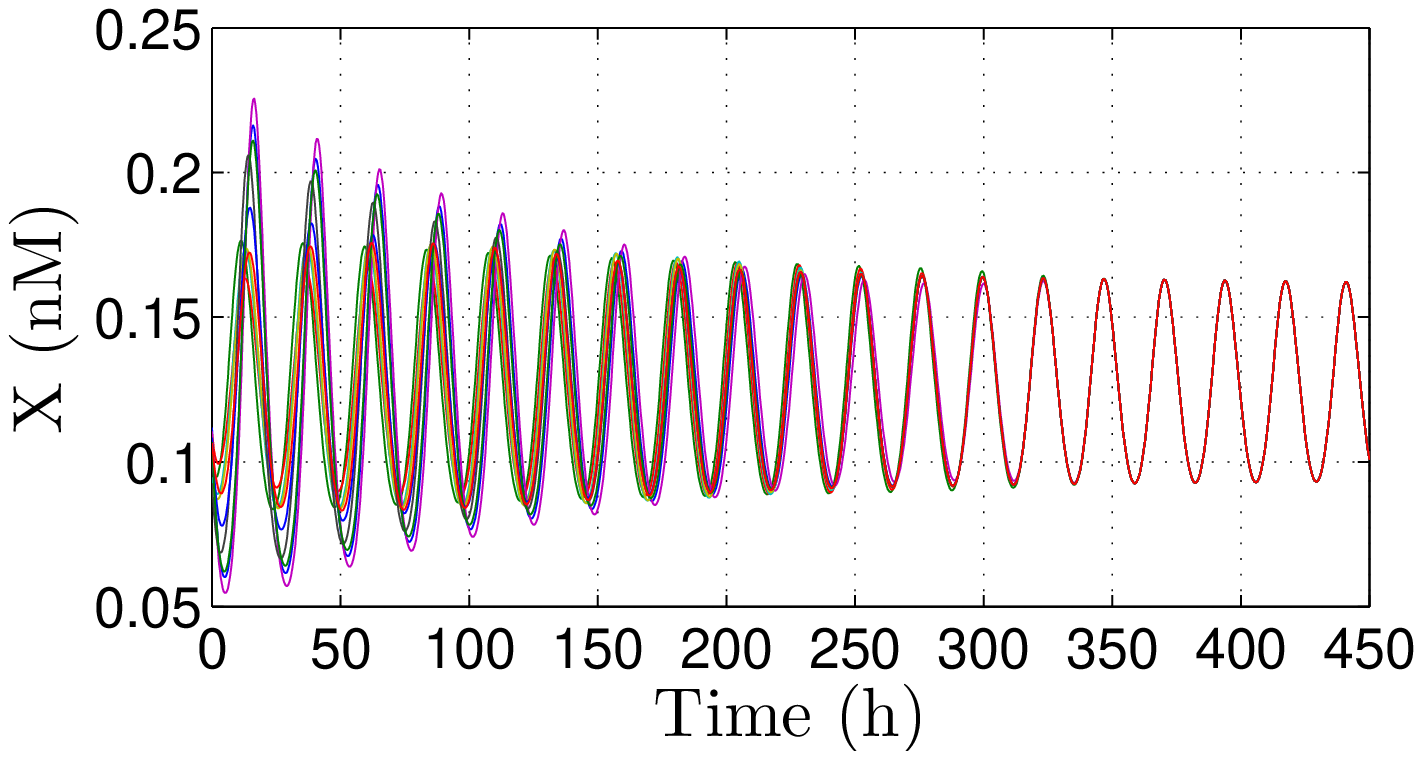}
\caption*{$c=100$}
\end{subfigure}
\caption{Dynamics of the mRNA levels ($X^i$) for $c=0,1,10,100$.}
\end{figure}

\section{Concluding remarks}\label{sec.conclude}

Theorem~\ref{thm-my} of this note has shown that a similar conclusion in comparison to that of \cite{HaStSeGo12} holds even when the coupled biochemical oscillators are under input constraints,
imposed by the requirements of biological feasibility and necessity to explicitly estimate the coupling gains. To satisfy these constraints, we
combine the linear coupling protocol from \cite{HaStSeGo12} with a ``saturating'' nonlinear block. We have proved that strong diffusive couplings can get coupled CFS-type oscillators synchronized when the saturation nonlinearity of the oscillators' inputs belongs to the identified class. Our proof is based on the synchronization criterion from \cite{HaStSeGo12}, extended to the systems with additional saturated block. The result may be extended to the CFS coupled through outputs $y_k$, as considered in \cite{HaraDoyle14}.
The techniques of quadratic constraints, used in our recent paper \cite{ProCaoZhangSch15}, allow to extend our results to some other types of ``saturated'' protocols, where not only control inputs,
but also outputs (or their deviations) are saturated.

The results of our paper can be applied e.g. to networks of \emph{synthetic} biochemical oscillators \cite{OBrien2012} where the couplings between the individual oscillators are artificially engineered. 
However, as has been reported by biochemists and biophysicists, the couplings between many natural biochemical oscillators, in particular neurons of the circadian pacemakers, are in general not diffusive 
\cite{GoBeWaKrHe05}. Hence, we are studying models for biochemical oscillators under mean field coupling \cite{GoBeWaKrHe05,LoWeKrHe08} or more complicated nearest-neighbor 
couplings \cite{VaHe11}, regulated by the concentrations of neurotransmitting polypeptides.

\bibliographystyle{unsrt}
\bibliography{ref_ming}

\end{document}